\documentclass{article}
\usepackage[utf8]{inputenc}
\usepackage[margin=1in]{geometry}
\usepackage{amssymb,amsfonts,amsmath,amsthm,comment,mleftright,enumerate,graphicx,hyperref,xcolor,algorithm2e,siunitx,diagbox}
\SetKwComment{Comment}{/* }{ */}
\usepackage[makeroom]{cancel}
\usepackage[title]{appendix}

\newtheorem{theorem}{Theorem}
\newtheorem{lemma}{Lemma}

\newcommand{\paren}[1]{\left(#1\right)}
\newcommand{\bracket}[1]{\left[#1\right]}
\renewcommand{\brace}[1]{\left\{#1\right\}}

\newcommand{\ceil}[1]{\left\lceil#1\right\rceil}

\newcommand{\F}{\mathbb{F}}
\newcommand{\R}{\mathbb{R}}

\newcommand{\M}[1]{\begin{bmatrix}#1\end{bmatrix}}
\newcommand{\MA}[2]{\bracket{\begin{array}{#1}#2\end{array}}}

\newcommand{\casew}[1]{\begin{cases}#1&\mathrm{else}\end{cases}}
\newcommand{\poly}{\mathrm{poly}}
\newcommand{\rank}[1]{\mathrm{rank}\paren{#1}}
\newcommand{\rowspan}[1]{\mathrm{rowspan}\paren{#1}}
\newcommand{\clmspan}[1]{\mathrm{clmspan}\paren{#1}}
\newcommand{\nullspace}[1]{\mathrm{nullspace}\paren{#1}}
\renewcommand{\dim}[1]{\mathrm{dim}\paren{#1}}
\newcommand{\rref}[1]{\mathrm{rref}\paren{#1}}
\newcommand{\GL}[2]{\mathrm{GL}\paren{#1,\ #2}}
\renewcommand{\~}[1]{\widetilde{#1}}
\renewcommand{\vec}[1]{\overrightarrow{#1}}
\renewcommand{\Vec}[1]{\mathrm{vec}\paren{#1}}
\renewcommand{\span}[1]{\mathrm{span}\paren{#1}}
\renewcommand{\min}[1]{\mathrm{min}\paren{#1}}
\renewcommand{\max}[1]{\mathrm{max}\paren{#1}}
\newcommand{\T}{\mathcal{T}}
\newcommand{\fexp}[2]{\mathcal{F}\paren{#1,\ #2}}
\newcommand{\pexp}[3]{\mathcal{P}\paren{#1,\ #2,\ #3}}

\title{Fixed-parameter tractability of canonical polyadic decomposition over finite fields}
\author{Jason Yang}
\date{}

\begin{document}

\maketitle

\begin{abstract}
We present a simple proof that finding a rank-$R$ canonical polyadic decomposition of a 3-dimensional tensor over a finite field $\F$ is fixed-parameter tractable with respect to $R$ and $\F$. We also show a nontrivial upper bound on the time complexity of this problem.
\end{abstract}

\section{Introduction}
Given a $n_0\times n_1\times n_2$ tensor (multidimensional array) $\T$, a rank-$R$ canonical polyadic decomposition (CPD) is a triplet of \textit{factor matrices} $A,B,C$ with shapes $R\times n_0,\ R\times n_1,\ R\times n_2$ respectively, such that (s.t.)

\[\forall 0\le i<n_0,\ 0\le j<n_1,\ 0\le k<n_2 : \T_{i,j,k}=\sum_{r=0}^{R-1} A_{r,i}B_{r,j}C_{r,k}.\]

This equation can also be written as $\T=\sum_{r=0}^{R-1} A_{r,:}\times B_{r,:}\times C_{r,:}$, where $M_{r,:}$ denotes the $r$-th row of $M$ and $\times$ denotes the outer product (defined in Section \ref{notation}).

CPD, also called parallel factor analysis (PARAFAC), plays a central role in the field of fast matrix multiplication, as finding an asymptotically fast algorithm for this problem is equivalent to finding a low-rank CPD of a specific 3-dimensional tensor \cite{survey}.
Although normally these CPDs should be over $\R$, there has been significant progress in recent years that restricts CPDs to be over a finite field, especially $\F_2$ \cite{alphatensor} \cite{flip} \cite{flip2} \cite{adaflip}, with some of these CPDs having been lifted to the integers (so that they also work over $\R$).
This trend motivates us to restrict our attention to CPDs over finite fields.

Our objective is to solve rank-$R$ CPD for an arbitrary tensor $\T$.
\textit{Solving} means that we must detect whether or not a rank-$R$ CPD of $\T$ exists, and return such a CPD if it exists.
We require an algorithm for this problem to give an exact (not approximate) CPD, to work for all possible $\T$, and to be guaranteed to be correct; to our knowledge, no previous work on CPD \cite{alphatensor} \cite{flip} \cite{smirnov} \cite{simul-diag} \cite{cpdqz} \cite{fpt-approx} satisfies all conditions at once, and most previous work only considers CPDs over $\R$.

We present a simple proof that solving rank-$R$ CPD is fixed-parameter tractable with respect to (w.r.t.) the rank $R$ and ground field $\F$. This result is formally stated as Theorem \ref{main} and proven in Section \ref{fpt}.
Additionally, this is close to the best result we can hope for, since if $R$ is allowed to vary linearly w.r.t. $n_0,n_1,n_2$, then rank-$R$ CPD is NP-Hard over any finite field \cite{hardness}.

\begin{theorem}
\label{main}
For whole number $R$ and finite field $\F$, determining whether a rank-$R$ CPD over $\F$ of a $n_0\times n_1\times n_2$ tensor exists, and constructing such a CPD if it exists, can be done in $O(f(|\F|,R)+\poly(n_0,n_1,n_2,R))$ time for some function $f$ and $O(\poly(n_0,n_1,n_2,R))$ space, where both $\poly(n_0,n_1,n_2,R)$ terms have $O(1)$ degree w.r.t. $|\F|$ and $R$.
\end{theorem}

We also give a nontrivial upper bound on $f$ in Section \ref{tighter}.

\subsection{Notation}
\label{notation}
\begin{itemize}
    \item We use NumPy notation for indexing on tensors.
    \begin{itemize}
        \item e.g. for a matrix $M$:
        \begin{itemize}
            \item $M_{i,:}$ denotes the $i$-th row (flattened into a 1D list by default, unless it is clear by context that it should not be flattened)
            \item $M_{:i,:}$ denotes the matrix truncated to its first $i$ rows
            \item $M_{:,i}$ denotes the $i$-th column (flattened into a 1D list by default)
        \end{itemize}
    \end{itemize}
    
    \item All tensors and sequences are 0-indexed.
    
    \item The outer product of two arbitrary tensors $X, Y$ is denoted as $X\times Y$ and defined as \[(X\times Y)_{i_0,\dots, i_{k-1}, j_0,\dots, j_{\ell-1}}
    :=X_{i_0,\dots, i_{k-1}}Y_{j_0,\dots, j_{\ell-1}}.\]
    
    \item The reduced row echelon form of a matrix $M$ is denoted as $\rref{M}$.
    
    \item A vector is \textit{canonically normalized} if it contains at least one nonzero element and its first (lowest-indexed) nonzero element is $1$.
    The \textit{canonical normalization} of $v\ne\vec{0}$ is defined as $v$ divided by its first nonzero element.

    \item A vector is \textit{monomial} if it contains exactly one nonzero element.

    \item Asymptotic complexity will sometimes be expressed as $\fexp{n}{k}:=O\paren{\frac{|\F|^n}{(|\F|-1)^k} R^{O(1)}}$, where the exponent on $R$ is $O(1)$ w.r.t. $n, k, |\F|, R$.
\end{itemize}

\section{Proof of fixed-parameter tractability}
\label{fpt}
Given an input tensor $\T\in\F^{n_0\times n_1\times n_2}$, construct a basis of its $i$-slices $\T_{i,:,:}$. Explicitly, we apply row reduction on the flattened tensor $\~{\T}:=\MA{c}{\vdots\\\hline \Vec{\T_{i,:,:}}\\\hline \vdots}_i$, creating a matrix $\Gamma\in\GL{n_0}{\F}$ s.t. $\Gamma\~{\T}=\rref{\~{\T}}$.
Removing the all-zero rows of $\rref{\~{\T}}$ and unflattening yields a $R_0\times n_1\times n_2$ tensor $\T'$, where $R_0=\rank{\~{T}}$, satisfying the property $\forall i':\sum_i \Gamma_{i',i} \T_{i,:,:}=\casew{\T'_{i',:,:}&i'<R_0 \\ O}$.

Using the invertibility of $\Gamma$, one can convert a rank-$R$ CPD for $\T'$ into one for $\T$ and vice versa, so solving rank-$R$ CPD for $\T$ is equivalent to solving rank-$R$ CPD for $\T'$:
\[
\forall i,j,k: \sum_{0\le r<R} A_{r,i}B_{r,j}C_{r,k}=\T_{i,j,k}
\Longrightarrow
\forall i',j,k: \sum_{0\le r<R} \paren{\sum_{0\le i<n_0} \Gamma_{i',i}A_{r,i}}B_{r,j}C_{r,k}=\T'_{i',j,k}
\]
\[
\forall i',j,k: \sum_{0\le r<R} A'_{r,i'}B_{r,j}C_{r,k}=\T'_{i',j,k}
\Longrightarrow
\forall i,j,k: \sum_{0\le r<R} \paren{\sum_{0\le i'<R_0} \paren{\Gamma^{-1}}_{i,i'}A'_{r,i'}}B_{r,j}C_{r,k}=\T_{i,j,k}.
\]

The upshot of this process is that the existence of a rank-$R$ CPD of $\T$ implies that $R_0\le R$,
since having $\T=\sum_{0\le r<R} A_{r,:}\times B_{r,:}\times C_{r,:}$ implies $\forall i: \T_{i,:,:}\in\span{\brace{B_{r,:}\times C_{r,:}}_{0\le r<R}}$.
Thus, if we detect $R_0>R$, we immediately know there is no solution and can terminate.

We can apply this basis extraction along the other two dimensions,
first constructing $\T''\in\F^{R_0\times R_1\times n_2}$ s.t. $\brace{\T''_{:,j',:}}_{j'}$ is a basis of $\brace{\T'_{:,j,:}}_j$,
then constructing $\T'''\in\F^{R_0\times R_1\times R_2}$ s.t. $\brace{\T'''_{:,:,k'}}_{k'}$ is a basis of $\brace{\T''_{:,:,k}}_k$.
Using a similar argument as above, we have $R_1\le R$ and $R_2\le R$ (otherwise we terminate), and solving rank-$R$ CPD of $\T'''$ is equivalent to solving rank-$R$ CPD for $\T$.

Since $\T'''$ has length $\le R$ across each dimension, we can simply use brute force to solve rank-$R$ CPD for it, without spending an amount of time that depends on $n_0,n_1,n_2$.
Doing so takes $\fexp{R(R_0+R_1+R_2)}{0}\in\fexp{3R^2}{0}$ time and $O(R(R_0+R_1+R_2))\in O(R^2)$ space.

Finally, constructing $\T',\T'',\T'''$ and converting a CPD for $\T'''$ into one for $\T$ can be done with several row reductions and matrix products, each of which takes polynomial time and space, so we are done.

One can extend this proof to show that rank-$R$ CPD over an arbitrary $D$-dimensional tensor is fixed-parameter tractable w.r.t. the rank $R$, ground field $\F$, and tensor dimension $D$.

\section{Tighter time complexity}
\label{tighter}
The proof in Section \ref{fpt} naturally corresponds to an algorithm for rank-$R$ CPD. Here we present tighter bounds on the time complexity for this algorithm. For simplicity, we sometimes upper-bound $n_0,n_1,n_2$ with $n:=\max{n_0,n_1,n_2}$ and $R_0,R_1,R_2$ with $R$ when calculating asymptotic complexity; we also assume $n\ge R\ge 1$.

\subsection{Pre- and post-processing}
Preprocessing, i.e. constructing $\T',\T'',\T'''$ and their corresponding change-of-basis matrices $\Gamma\in\GL{n_0}{\F}$, $\Gamma'\in\GL{n_1}{\F}$, $\Gamma''\in\GL{n_2}{\F}$, can be done by row-reducing flattenings of $\T$, $\T'$, and $\T''$, which have shapes $n_0\times n_1n_2$, $n_1\times R_0n_2$, and $n_2\times R_0R_1$ respectively.

Instead of na\"ively running row-reduction, we notice that our algorithm only has to run to completion if each of the flattenings of $\T$, $\T'$, and $\T''$ have rank $\le R$. We can thus add early termination to row-reduction to get the following result:

\begin{lemma}
\label{fast-rankfac}
Given a matrix $M\in\F^{m\times n}$, and integer $r'$, we can do the following in $O((\min{r,r'}+1)(m+1)(m+n))$ time, where $r:=\rank{M}$:
\begin{itemize}
    \item If $r\le r'$: return some $C,J\in\GL{m}{\F}$ and $F\in\F^{r\times n}$ s.t. $M=C_{:,:r}F$ and $J=C^{-1}$.
    \item Else: return nothing
\end{itemize}
\end{lemma}
We give a proof in Appendix \ref{fast-rankfac-proof}. Using this result allows preprocessing to be done in $O((R+1)n^3)$ time.

Postprocessing, i.e. converting a rank-$R$ CPD of $\T'''$ (if it exists) to a rank-$R$ of $\T$, involves multiplying each factor matrix of the CPD of $\T'''$ with $\le R$ many columns of the corresponding change-of-basis matrix (e.g. when converting from $\T'$ to $\T$ in the previous section, we calculate $A=A'(\Gamma^{-1}_{:,:R_0})^\intercal$), so this step takes $O(R^2n)$ time.

Thus, we can replace the $\poly(n_0,n_1,n_2,R)$ term in Theorem \ref{main} with $O((R+1)n^3)$; for fixed $R$, this is linear w.r.t. the total input size (i.e. the number of elements in $\T$) if $\T$ has shape $n\times n\times n$.

\subsection{CPD solving}
The bottleneck of our algorithm is solving rank-$R$ CPD $\T'''=\sum_r A_{r,:}\times B_{r,:}\times C_{r,:}$, where $\T'''\in\F^{R_0\times R_1\times R_2}$.

We can improve over brute force by only fixing $A$ and $B$, and then solving for $C$ via a linear system.
WLOG we can also force each row of $A$ and $B$ to be canonically normalized, due to the scale-invariance $a\times b\times c=(\alpha a)\times (\beta b)\times (\frac{1}{\alpha\beta} c)$, yielding running time $\fexp{R(R_0+R_1)}{2R}\in\fexp{2R^2}{2R}$ \footnote{We can also force pairs of row-vectors $(A_{r,:},B_{r,:})$ to be lexicographically nondecreasing w.r.t. $r$, but this improves running time by at most a factor of $R!$, which is independent of $|\F|$, so we ignore this result. In practice this optimization is important.}.

In the rest of this section, we obtain a slightly lower time complexity by fixing $A$ and solving $B, C$ in aggregate.
WLOG we can force each row of $A$ to be canonically normalized \footnote{Similarly, we can also force rows $A_{r,:}$ to be lexicographically nondecreasing w.r.t. $r$, but we ignore this speedup.}.

By defining $M_r:=B_{r,:}\times C_{r,:}$, our CPD equation is equivalent to solving the system $\forall i:\sum_r A_{r,i}M_r=\T'''_{i,:,:}$ for matrices $M_r\in\F^{R_1\times R_2}$ under the constraint that $\forall r:\rank{M_r}\le 1$ \footnote{After finding all $M_r$, a valid assignment for $B$ and $C$ can be found using matrix rank factorization.}.
By abuse of notation we abbreviate this system as $A^\intercal \MA{c}{\vdots \\ M_r \\ \vdots}_r = \MA{c}{\vdots \\ \T'''_{i,:,:} \\ \vdots}_i$, where the $M_r$ and $\T'''_{i,:,:}$ are not concatenated, just grouped into column vectors (notice the absence of horizontal bars).

\subsubsection{Row reduction}
Performing row operations on this system is equivalent to left-multiplying both sides of the system by an invertible matrix $S\in\GL{R_0}{\F}$, and doing so does not change whether or not a solution exists. Formally, left-multiplying by $S$ transforms the system into $(SA^\intercal)\MA{c}{\vdots \\ M_r \\ \vdots}_r=\MA{c}{\vdots \\ D_i \\ \vdots}_i$, where $D_i:=\sum_j S_{i,j}\T'''_{j,:,:}$.

We want to choose $S$ s.t. $SA^\intercal$ has a lot of monomial columns \footnote{Since every row of $A$ to be nonzero and $S$ is invertible, every column of $SA^\intercal$ must contain at least one nonzero element, so we can ignore all-zeros columns.}; this is because when we fix each $M_r$ s.t. $(SA^\intercal)_{:,r}$ is \textit{not} monomial, the system reduces to a list of independent equations,
allowing all other $M_r$ to be trivially solved (or the nonexistence of a solution to be trivially detected).

For arbitrary $S$, consider a subset of column indices $j_0,\dots,j_{K-1}$ s.t. $\forall i:(SA^\intercal)_{:,j_i}$ is monomial.
Because left-multiplying by $S$ is a linear transformation over $\F^{R_0}$, we have that if the nonzero elements of columns $(SA^\intercal)_{:,j_i}$ are the same, then all columns $A^\intercal_{:,j_i}$ must be scalar multiples of a common vector; and if if the nonzero elements of columns $(SA^\intercal)_{:,j_i}$ are mutually distinct, then all columns $A^\intercal_{:,j_i}$ must be linearly independent.
Thus, choosing optimal $S$ is equivalent to finding a linearly independent subset of vectors $V\in\F^{R_0}$ that maximizes $\paren{\# \ j \textrm{ s.t. } A_{:,j}\in\bigcup_{v\in V} \span{\brace{v}}}$.
We present a greedy heuristic to choose moderately good $S$:

\begin{algorithm}
    \caption{Greedy heuristic for maximizing number of monomial columns,
    denoted as \texttt{greedy}$(W)$
    }
    \label{monomial-greedy}
    \KwData{
        $W\in\F^{m\times n}$ s.t. $\forall j:W_{:,j}\ne\vec{0}$
    }
    \KwResult{$S\in\GL{m}{\F}$ s.t. $SW$ has a large number of monomial columns}
    $U\gets\brace{\brace{\eta(W_{:,j}):0\le j<n}}$ \Comment*[r]{multiset of vectors
    \\
    $\eta(v)$ denotes the canonical normalization of $v$.}
    $V\gets\emptyset$ \Comment*[r]{subset of columns of $W$}
    \While{$\exists v\in U$ s.t. $v\not\in\span{V}$}{
        $v^*\gets$ element in $U$ with highest multiplicity s.t. $v^*\not\in\span{V}$\;
        $V\gets V\cup\brace{v^*}$
    }
    $M_V\gets\MA{c|c|c}{\cdots&v_i&\cdots}$ for each element $v_i\in V$, in arbitrary order\;
    $S\gets$ matrix s.t. $SM_V=\rref{M_V}=\MA{c}{I_{|V|}\\\hline O}$ (from Lemma \ref{fast-rankfac})\;
    \Return{$S$}
\end{algorithm}

It is clear that calculating $\mathtt{greedy}(A^\intercal)$ runs in $\poly(R)$ time, since checking if an arbitrary vector $u$ is in $\span{V}$ can be done with Gaussian elimination.
In Appendix \ref{monomial-greedy-proof}, we prove that by setting $S\gets\texttt{greedy}(A^\intercal)$, $SA^\intercal$ is guaranteed to have at least $\ceil{K\max{1,\frac{R}{\paren{\frac{|\F|^K-1}{|\F|-1}}}}}$ many monomial columns, where $K:=\rank{A^\intercal}=\rank{A}$; we also prove that this is the best possible worst-case lower bound.

\subsubsection{Refinement of matrix factorization}
We can further reduce the number of $M_r$ we enumerate by one (and have extra speedup) by using some casework on matrix ranks and properties of matrix factorization.

Set $S$ to $\texttt{greedy}(A^\intercal)$ using Algorithm \ref{monomial-greedy}.
Because the vector subset $V$ in the algorithm will be a basis of $\clmspan{A^\intercal}$ at the end of the algorithm, the bottom $R_0-K$ rows of $SA^\intercal$ must be all-zeros, so WLOG we can omit them \footnote{If $\forall K\le i<R_0:D_i=O$, then these rows have no effect; otherwise, no solution exists.}.
By relabeling variables, the system is equivalent to solving

\[\forall 0\le i<K: X_i+\sum_{0\le p<P} \alpha_{i,p} Y_p=D_i\]

\[\textrm{for matrices } X_i,Y_p\in\F^{R_1\times R_2},\]

\[\textrm{subject to } \forall i: \rank{X_i}\le \chi_i \textrm{ and } \forall p: \rank{Y_p}\le 1,\]

where $\chi_i:=(\textrm{\# of monomial columns in } SA^\intercal \textrm{ with nonzero element at row } i)$,
$P:=(\textrm{\# of non-monomial} \\ \textrm{ columns in } SA^\intercal)$, and $\alpha_{i,p}\in\F$ are select elements of $SA^\intercal$.
By construction, we have that $\forall i:\chi_i\ge 1$ and $P+\sum_i \chi_i = R$. As a consequence, $\forall i:\chi_i\le \overline{\chi}$, where $\overline{\chi}:=R-P-K+1$.

A necessary condition for a solution to exist is that $\rank{D_i}\le n_i$, where $n_i:=\chi_i+(\# \textrm{ of } p \textrm{ s.t. } \alpha_{i,p}\ne 0)$; WLOG we assume this condition from this point onward.

We then solve the system by breaking into cases over $P$:

\noindent \textbf{If $\boldsymbol{P=0}$}: the system can trivially be solved in $\fexp{0}{0}$ time.

\noindent \textbf{If $\boldsymbol{P=1}$}: we use the following lemma about matrix rank factorizations, which we prove in Appendix \ref{fullrank-enum-proof}:

\begin{lemma}
\label{fullrank-enum}
For $M\in\F^{m\times n}$ with rank $r$, consider $r$ many matrices $M_0,\dots,M_{r-1}\in\F^{m\times n}$ s.t. $\rank{M_i}\le 1$ and $M=\sum_{0\le i<r} M_i$.
Then for any $i$, $M_i$ has $\fexp{2r}{1}$ many possibilities, and it takes that much time to enumerate them.
\end{lemma}

To solve our system, we break into two subcases:
\begin{itemize}
    \item If $\exists i$ s.t. $\rank{D_i}=n_i$ and $n_i=\chi_i+1$:

    we have $\alpha_{i,0}\ne 0$, so we can enumerate all $\fexp{2n_i}{1}\le \fexp{2(\chi_i+1)}{1}$ many possibilities of $Y_0$.

    \item Else:
    
    we have that $\forall i:\rank{D_i}\le \chi_i$, so WLOG we can set $Y_0=O$.
\end{itemize}

For each assignment of $Y_0$ we fix, we solve the rest of the system similarly to having $P=0$. In the worst case, solving the $P=1$ case takes $\fexp{2(\overline{\chi}+1)}{1}$ time.

\noindent \textbf{If $\boldsymbol{P=2}$}: we use the following lemma about matrix factorizations containing one more rank than necessary, which we prove in Appendix \ref{add1rank-enum-proof}:

\begin{lemma}
\label{add1rank-enum}
For $M\in\F^{m\times n}$ with rank $r-1$, consider $r$ many matrices $M_0,\dots,M_{r-1}\in\F^{m\times n}$ s.t. $\rank{M_i}\le 1$ and $M=\sum_{0\le i<r} M_i$.
Then for any $i$, $M_i$ has $\fexp{m+r-1}{1}+\fexp{n+r-1}{1}$ many possibilities, and it takes that much time to enumerate them.
\end{lemma}

To solve our system, we break into three subcases:
\begin{itemize}
    \item If $\exists i$ s.t. $\rank{D_i}=n_i$ and $n_i\ge \chi_i+1$:

    there must be at least one $p$ s.t. $\alpha_{i,p}\ne 0$, so for such $p$ we can enumerate all $\fexp{2n_i}{1}\in \fexp{2(\chi_i+2)}{1}$ many possibilities of $Y_p$.

    \item Else if $\exists i$ s.t. $\rank{D_i}=n_i-1$ and $n_i=\chi_i+2$:

    we have $\alpha_{i,0}\ne 0$ and $\alpha_{i,1}\ne 0$, so we can enumerate all $\fexp{R_1+n_i-1}{1}+\fexp{R_2+n_i-1}{1}\in \fexp{R+\chi_i+1}{1}$ many possibilities of $Y_1$.

    \item Else:

    we have that $\forall i:\rank{D_i}\le \chi_i$, so WLOG we can set $Y_p=O$ for all $p$.
\end{itemize}

For each assignment of some $Y_p$ we fix, we solve the rest of the system similarly to having $P$ be lower. In the worst case, this adds an extra $\fexp{R+\overline{\chi}+1}{1}$ multiplicative factor to the time complexity for the $P=1$ case.


\noindent \textbf{If $\boldsymbol{P>2}$}: we enumerate $Y_2,\dots,Y_{P-1}$ to reduce the system to the $P=2$ case, adding an extra $\fexp{2R}{1}^{P-2}$ multiplicative factor to the time complexity for the $P=1$ case.

Thus, for general $P$ the system can be solved in \[
\begin{array}{rl}
\pexp{R}{K}{P}&:=\casew{
    \fexp{0}{0} & P=0 \\
    \fexp{2(\overline{\chi}+1)}{1} & P=1 \\
    \fexp{2(\overline{\chi}+1)}{1}\cdot\fexp{R+\overline{\chi}+1}{1}\cdot \fexp{2R}{1}^{P-2}
}
\\
&=\casew{
    \fexp{0}{0} & P=0 \\
    \fexp{2(\overline{\chi}+1)}{1} & P=1 \\
    \fexp{R(2P-3)+3\overline{\chi}+3}{P}
}
\\
&=\casew{
    \fexp{0}{0} & P=0 \\
    \fexp{2(R-K+1)}{1} & P=1 \\
    \fexp{2PR-3(P+K)+6}{P}
}
\end{array}
\] time.

\subsubsection{Time complexity}
To bound the total running time of our algorithm, we first bound the number of matrices of fixed rank with normalized rows.
\begin{lemma}
\label{mat-count}
There are $\fexp{r(m+n-r+1)}{r+m}$ many matrices $A\in\F^{m\times n}$ s.t. $\rank{A}=r$ and each row of $A$ is canonically normalized.
\end{lemma}
\begin{proof}
By Lemma \ref{fullrank}, $A$ has $|\GL{r}{\F}|$ many rank factorizations $(U\in\F^{m\times r},\ V\in\F^{r\times n})$. Since different matrices obviously cannot have the same factorization in common, the number of distinct $A\in\F^{m\times n}$ with rank $r$ is $\frac{|\F|^{mr+rn}}{|\GL{r}{\F}|}$.
If we furthermore restrict each row of $A$ to be nonzero, then each row can be independently replaced with $|\F|-1$ many nonzero scalar multiples, so the number of rank-$r$ matrices $A$ with canonically normalized rows is $\le \frac{|\F|^{mr+rn}}{|\GL{r}{\F}|\cdot (|\F|-1)^m}$.
Since $|\GL{r}{\F}|=\prod_{i=0}^{r-1} (|\F|^r-|\F|^i)\ge (|\F|^r-|\F|^{r-1})^r=|\F|^{r(r-1)}(|\F|-1)^r$, $\frac{|\F|^{mr+rn}}{|\GL{r}{\F}|\cdot (|\F|-1)^m}
\in \fexp{r(m+n-r+1)}{r+m}$.
\end{proof}


By summing over all $A$ that we fix, the total running time of our algorithm is upper-bounded by $\sum_{K=1}^R \fexp{K(2R-K+1)}{K+R} \cdot \pexp{R}{K}{R-\mu(R,K)}$ \footnote{Note that we do not have to consider $K=0$, since each row in $A$ is required to not be all-zeros.}, where $\mu(R,K):=\ceil{K\max{1,\frac{R}{\paren{\frac{|\F|^K-1}{|\F|-1}}}}}$.

Because we assume $R\ge 1$, we have $\mu(R,1)=R$, so the $K=1$ term in the summation can be replaced with $\fexp{2R}{R+1}$ \footnote{This upper bound provided by Lemma \ref{mat-count} is quite loose for matrix rank $r=1$, as all allowed matrices must equal $\M{1\\\vdots\\1}v$ for canonically normalized row vector $v$, yielding $\fexp{n}{1}$ many such matrices instead of $\fexp{m+n}{m+1}$. This discrepancy does not affect our final result.}.

Since the summation is difficult to further simplify, we give a loose upper bound by replacing $\mu(R,K)$ with $K$, which is accurate for large $|\F|$:
then 
the asymptotic running time simplifies to
\[\begin{array}{ll}
&\fexp{2R}{R+1} + \sum_{K=2}^R \fexp{K(2R-K+1)}{K+R} \cdot \pexp{R}{K}{R-K} \\
\\
=&\fexp{2R}{R+1} \\
&+\sum_{K=2}^{R} \fexp{K(2R-K+1)}{K+R} \cdot \casew{
    \fexp{0}{0} & R-K=0 \\
    \fexp{4}{1} & R-K=1 \\
    \fexp{2(R-K)R-3R+6}{R-K}
} \\
\\
=&\fexp{2R}{R+1} \\
&+\sum_{K=2}^{R} \casew{
    \fexp{R(R+1)}{2R} & K=R \\
    \fexp{R^2+R+2}{2R} & K=R-1 \\
    \fexp{2R^2-3R-K^2+K+6}{2R}
} \\
=&\fexp{2R}{R+1} \\
&+\casew{
\fexp{2R^2-3R+4}{2R} & R\ge 4 \\
\fexp{14}{6} & R=3 \\
\fexp{6}{4} & R=2 \\
0
} \\
\\
=&\fexp{\paren{\begin{cases}
2R^2-3R+4 & R\ge 4 \\
14 & R=3 \\
6 & R=2 \\
2 & R=1
\end{cases}}}{2R}. \\
\end{array}\]

For small $|\F|$, we avoid replacing $\mu(R,K)$, and we estimate the running time constant factor by replacing $\fexp{n}{k}$ with $f(n,k):=\frac{|\F|^n}{(|\F|-1)^k}$ and summing up terms instead of extracting the largest term.
Table \ref{results} summarizes the approximate constant factor in asymptotic running time of our algorithm versus the technique of fixing two factor matrices.

\begin{table}[h]
    \centering
    \begin{tabular}{cc}
        \begin{tabular}{|c|c|c|c|}
            \hline
            \diagbox{$R$}{$\F$} & $\F_2$ & $\F_3$ & large \\
            \hline
            1 & $\num{4.0}$ & $\num{2.25}$ & $f(2,2)$ \\
            \hline
            2 & $\num{80.0}$ & $\num{55.7}$ & $f(6,4)$ \\
            \hline
            3 & $\num{2.05e+04}$ & $\num{8.31e+04}$ & $f(14,6)$ \\
            \hline
            4 & $\num{6.29e+06}$ & $\num{1.24e+09}$ & $f(24,8)$ \\
            \hline
            5 & $\num{3.98e+10}$ & $\num{5.45e+13}$ & $f(39,10)$ \\
            \hline
        \end{tabular}
        &
        \begin{tabular}{|c|c|c|c|}
            \hline
            \diagbox{$R$}{$\F$} & $\F_2$ & $\F_3$ & large \\
            \hline
            1 & $\num{4.0}$ & $\num{2.25}$ & $f(2,2)$ \\
            \hline
            2 & $\num{2.56e+02}$ & $\num{4.1e+02}$ & $f(8,4)$ \\
            \hline
            3 & $\num{2.62e+05}$ & $\num{6.05e+06}$ & $f(18,6)$ \\
            \hline
            4 & $\num{4.29e+09}$ & $\num{7.24e+12}$ & $f(32,8)$ \\
            \hline
            5 & $\num{1.13e+15}$ & $\num{5.91e+15}$ & $f(50,10)$ \\
            \hline
        \end{tabular}
    \end{tabular}
    
    \caption{Approximate constant factor $C$ for solving rank-$R$ CPD of a $R\times R\times R$ tensor over finite field $\F$ in $O(C\poly(R))$ time, by using our algorithm (left) and by fixing two factor matrices and solving the third via a linear system (right).
    The ``large" column represents arbitrarily large finite fields, and the function $f(n,k)$ denotes $\frac{|\F|^n}{(|\F|-1)^k}$.}
    \label{results}
\end{table}

\section{Future directions}
We conclude with some open questions, which are natural generalizations of our work:

\begin{itemize}
    \item What is the lowest possible constant factor for solving rank-$R$ CPD of a 3-dimensional tensor over a finite field $\F$, for fixed $R$ and fixed $\F$?
    \item Is CPD fixed-parameter tractable over number systems beyond finite fields, such as the integers?
\end{itemize}

\section*{Acknowledgments}
We thank Erik Demaine and Ani Sridhar for inspiring us to start and continue researching this topic.

\newpage
\begin{appendices}
\section{}
\subsection{Fast low-rank matrix factorization}
\label{fast-rankfac-proof}

To prove Lemma \ref{fast-rankfac}, we run Gaussian elimination with early termination while maintaining extra information:
\begin{algorithm}
    \caption{Gaussian elimination with early termination by rank}
    \KwData{
        $M\in\F^{m\times n}$, integer $r'$
    }
    \KwResult{
    If $r:=\rank{M}\le r'$:
    $(C\in\GL{m}{\F},\ J\in\GL{m}{\F},\ F\in\F^{r\times n})$ s.t. $F=\rref{M}_{:r,:},\ M=C_{:,:r}F,\ J=C^{-1}$;
    else, $\emptyset$}
    $C\gets I_m$, $J\gets I_m$, $F\gets M$, $q\gets 0$\;

    \For{$j=0,\dots,n-1$}{
        \If{$\exists$ $q\le i<m$ s.t. $F_{i,j}\ne 0$}{
            $i\gets$ arbitrary index s.t. $F_{i,j}\ne 0$\;
            
            swap rows $i$ and $q$ of $F$\;
            swap rows $i$ and $q$ of $J$\;
            swap columns $i$ and $q$ of $C$\;

            $\sigma\gets F_{q,j}$ \Comment*[r]{A}
            $F_{q,:} \gets F_{q,:}/\sigma$ \Comment*[r]{A}
            $J_{q,:} \gets J_{q,:}/\sigma$ \Comment*[r]{A}
            $C_{:,q} \gets \sigma C_{:,q}$ \Comment*[r]{A}
            
            \For{$k=q+1,\dots,m-1$}{
                $s\gets F_{k,j}$ \Comment*[r]{B}
                $F_{k,:}\gets F_{k,:}-s F_{q,:}$ \Comment*[r]{B}
                $J_{k,:}\gets J_{k,:}-s J_{q,:}$ \Comment*[r]{B}
                $C_{:,q}\gets C_{:,q}+s C_{:,k}$ \Comment*[r]{B}
            }

            $q\gets q+1$\;
            \If{$q>r'$}{
                \Return{$\emptyset$}
            }
        }
    }

    \For{$i=q-1,\dots,0$}{
        $j\gets$ smallest index s.t. $F_{i,j}\ne 0$ \Comment*[r]{must have $F_{i,j}=1$ after running this instruction}
        \For{$k=0,\dots,i-1$}{
            $s\gets F_{k,j}$ \Comment*[r]{C}
            $F_{k,:}\gets F_{k,:}-s F_{q,:}$ \Comment*[r]{C}
            $J_{k,:}\gets J_{k,:}-s J_{q,:}$ \Comment*[r]{C}
            $C_{:,q}\gets C_{:,q}+s C_{:,k}$ \Comment*[r]{C}
        }
    }
    
    \Return{$(C,J,F_{:q,:})$}
\end{algorithm}

            

            

If the early termination condition $q>r'$ is never satisfied, then after the algorithm finishes, $q$ will equal $r:=\rank{M}$ and $F$ will equal $\rref{M}$. Thus:

\begin{itemize}
    \item if early termination never occurs, i.e. if $r\le r'$:
    \begin{itemize}
        \item the for-loops $(k=q+1,\dots,m-1)$ and $(k=0,\dots,i-1)$ are each started $r$ many times;
        \item $\Rightarrow$ each line labeled $\texttt{A}$, $\texttt{B}$, and $\texttt{C}$ is reached $r$, $O(mr)$, and $O(r^2)$ times respectively;
        \item $\Rightarrow$ the labeled lines contribute total time complexity $O((r+mr+r^2)(m+n))\in O((r+mr)(m+n))$ (since $r\le\min{m,n}$)
    \end{itemize}

    \item else:
    \begin{itemize}
        \item the for-loops $(k=q+1,\dots,m-1)$ and $(k=0,\dots,i-1)$ are each started $r'$ and 0 many times respectively;
        \item $\Rightarrow$ each line labeled $\texttt{A}$, $\texttt{B}$, and $\texttt{C}$ is reached $r'$, $O(mr')$, and 0 times respectively;
        \item $\Rightarrow$ the labeled lines contribute total time complexity $O((r'+mr')(m+n))$
    \end{itemize}
\end{itemize}

Overall, the labeled lines contribute total time complexity $O(\min{r,r'}(1+m)(m+n))$.

Finally, initializing $C,J,F$ takes $O(m^2+mn)$ time, so the total time complexity of the algorithm overall is $O((\min{r,r'}+1)(m+1)(m+n))$.

For correctness, the algorithm maintains the invariants $CJ=I_m$ and $JM=F$ after each group of lines with the same label finishes running one iteration, so at the end of the algorithm (if early termination never occurs) we satisfy $J=C^{-1}$ and $M=CF=C_{:,:q}F_{:q,:}$.


\subsection{Row reduction for monomial columns}
\label{monomial-greedy-proof}
For some $W\in\F^{m\times n}$ that does not contain all-zeros columns, we want to choose $S\in\GL{m}{\F}$ s.t. $SW$ has a large (not necessarily largest) number of monomial columns.

As in Algorithm \ref{monomial-greedy}, define $U$ as the multiset of canonical normalizations of the columns of $W$.
Let $V^*$ be the value of the vector subset $V$ at the end of the algorithm, and let $K=|V^*|$. The number of monomial columns in $SW$ is then equal to $\sum_{v\in V^*} m_U(v)$, where $m_U(v)$ is the multiplicity of $v$ in $U$.

Since the main loop only stops when $\span{V}$ contains all distinct elements of $U$, and each iteration increases $\dim{\span{V}}$ by exactly 1, we have that $\span{V^*}=\clmspan{W}$, so $K=\rank{W}$ and $U$ contains at most $N:=\frac{|\F|^K-1}{|\F|-1}$ many distinct vectors.
Let $\ell_0\ge\dots\ge \ell_{N-1}\ge 0$ be the list of $m_U(v)$ for every possible canonically normalized vector $v\in\span{V^*}$, sorted in nonincreasing order. By construction, $\sum_i \ell_i = |U|=n$.

During Algorithm \ref{monomial-greedy}, at the start of the $k$-th iteration of the main loop (where $k$ starts at 0), we have $|V|=k$, so $\span{V}$ contains $\frac{|\F|^k-1}{|\F|-1}$ many canonically normalized vectors. In the worst case, these vectors cover exactly that many distinct elements in $U$ of the highest multiplicity, so for the vector $v^*$ that we choose to add to $V$ in this iteration, we have $m_U(v^*)\ge \ell_{\frac{|\F|^k-1}{|\F|-1}}$. Additionally, since $v^*$ must exist in $U$, we have $m_U(v^*)\ge 1$.

Thus, the number of monomial columns in $SW$ is $\ge \max{K,\ \sum_{0\le k<K} \mu_k}$, where $\mu_k := \ell_{\frac{|\F|^k-1}{|\F|-1}}$.
To lower-bound this expression w.r.t. $n$, we can find an upper bound on $n$ w.r.t. $\mu_k$.
By replacing each $\ell_i$ in the summation $\sum_i \ell_i$ with $\mu_k$, where $k$ is the largest integer s.t. $\frac{|\F|^k-1}{|\F|-1}\le i$, we have $n\le \sum_{0\le k<K} |\F|^k \mu_k$.

Splitting the summation into suffix sums $S_k:=\sum_{k\le j<K} \mu_j$, we get $n\le S_0+\sum_{1\le j<K} \paren{|\F|^j-|\F|^{j-1}} S_j$.
Since $\mu_0\ge\mu_1\ge\dots\ge\mu_{K-1}$, we can upper bound $S_k$ w.r.t. $S_0$ by noticing that $S_k\le (K-k)\mu_k$ and $S_0 \ge S_k+k\ell_k$, so $S_k\le \frac{K-k}{K} S_0$.
Using this bound yields $n\le S_0 M$, where $M:=\paren{1+\sum_{1\le j<K} \paren{|\F|^j-|\F|^{j-1}}\cdot\frac{K-j}{K}}$; note that the number of monomial columns in $SW$ is at least $S_0$.

Calculating $M$ yields
\[
\begin{array}{ll}
M &=1+\frac{|\F|-1}{K}\sum_{1\le j<K} |\F|^{j-1} \paren{K-j}
\\
&=1+\frac{|\F|-1}{K}\sum_{1\le j\le i<K} |\F|^{j-1}
\\
&=1+\frac{|\F|-1}{K}\sum_{1\le i<K} \frac{|\F|^i-1}{|\F|-1}
\\
&=1+\frac{1}{K}\sum_{1\le i<K} (|\F|^i-1)
\\
&=1+\frac{1}{K}\paren{\frac{|\F|^K-|\F|}{|\F|-1} - (K-1)}
\\
&=\frac{1}{K}\paren{\frac{|\F|^K-|\F|}{|\F|-1} +1}
\\
&=\frac{1}{K} \frac{|\F|^K-1}{|\F|-1}
,
\end{array}
\]




and thus $SW$ has at least $\max{K,\ n \big/ \paren{\frac{1}{K} \frac{|\F|^K-1}{|\F|-1}}}
=K\max{1,\ n \big/ \paren{\frac{|\F|^K-1}{|\F|-1}}}$ many monomial columns. Since the number of monomial columns in $SW$ must be an integer, we can apply the ceiling function to this expression.
This lower bound is tight when the columns of $W$ contain every possible canonically normalized $n$-vector with equal multiplicity.

\subsection{Matrix factorization}
Define a \textit{rank-$r$ factorization} of a matrix $M\in\F^{m\times n}$ as a pair $(U\in\F^{m\times r},\ V\in\F^{r\times n})$ s.t. $M=UV$. If such a pair exists, $M$ does not necessarily have rank $r$; if that happens to be true, we say that $(U,V)$ is also a \textit{full-rank} factorization.

Note that rank-$r$ factorization is equivalent to rank-$r$ CPD of a 2D tensor, since for vectors $u_0,\dots,u_{r-1}$ and $v_0,\dots,v_{r-1}$, $\sum_i u_i\times v_i = \MA{c|c|c}{&\uparrow&\\\cdots&u_i&\cdots\\&\downarrow&}_i \MA{ccc}{&\vdots& \\ \hline \leftarrow&v_i&\rightarrow \\ \hline &\vdots&}_i$.
As a corollary, there is a bijection between rank-$r$ factorizations and summations $M=\sum_i M_i$ for $M_0,\dots,M_{r-1}\in\F^{m\times n}$ where $\forall i: \rank{M_i}\le 1$, since a matrix is equal to some outer product of two vectors if and only if it has rank $\le 1$.

To prove Lemmas \ref{fullrank-enum} and \ref{add1rank-enum}, we first show the uniqueness of rref up to all-zeros rows, then use this to characterize all full-rank factorizations of an arbitrary matrix.

\begin{lemma}
\label{unique-rref}
Two rref matrices with the same rowspan must be equal, up to the presence of all-zero rows.
\end{lemma}
\begin{proof}
Consider an rref matrix $W=\MA{c}{w_0\\\hline \vdots \\\hline w_{k-1} \\\hline O}$, where each $w_i\ne\vec{0}$ has its leading 1 at column $\ell_i$, with $\ell_0<\dots<\ell_{k-1}$. Consider a linear combination of the rows, $v:=\sum_i \alpha_i w_i$, where the $\alpha_i$ are not all zero: then $v_{\ell_i}=\alpha_i$, since by definition of rref, every column of $W$ that has a leading 1 has all of its other elements set to 0.
As a consequence, if we were only given $v$, then we could immediately determine the list of coefficients $\alpha_i$ of that linear combination reading each element $v_{\ell_i}$.

Additionally, if $v\ne\vec{0}$, the first nonzero element of $v$ must be at some $\ell_i$-th column, as we can use the following casework: if $\alpha_0\ne 0$, the leading nonzero term of $v$ is $\ell_0$;  else if $\alpha_1\ne 0$, the leading nonzero term of $v$ is $\ell_1$; etc.

Now suppose there is another rref matrix $W'$ s.t. $\rowspan{W}=\rowspan{W'}$: then $W'$ cannot have a leading 1 at a column that is not one of the $\ell_i$-th, otherwise $\rowspan{W'}$ would contain a row-vector that is not in $\rowspan{W}$. Conversely, $W'$ must contain a leading 1 at every $\ell_i$-th column, since $w_i\in\rowspan{W}$ for each $i$.
Thus, for each $0\le i<k$, $W'_{i,:}$ has a leading 1 at column $\ell_i$.

Finally, since $w_i\in\rowspan{W'}$ and, by definition of rref, $w_i$ contains a 1 at column $\ell_i$ and a 0 at all other leading columns $\ell_j$ for $j\ne i$, the only way for $w_i$ to be a linear combination of the rows of $W'$ is if $w_i=W'_{i,:}$. Thus, $W_{:k,:}=W'_{:k,:}$.
\end{proof}

\begin{lemma}
\label{fullrank}
For $M\in\F^{m\times n}$, let $r:=\rank{M}$.
Then there exists a rank-$r$ factorization $(C_0,\ F_0)$ of $M$.
Furthermore, any rank-$r$ factorization $(C,\ F)$ of $M$ must satisfy $(C,\ F)=(C_0X,\ X^{-1} F_0)$ for some $X\in\GL{r}{\F}$.
\end{lemma}
\begin{proof}
By Lemma \ref{fast-rankfac}, $\exists T\in\GL{m}{\F}$ s.t. $TM=\rref{M}$, so $T^{-1}\rref{M}=M$. Since all except the top $r$ rows of $\rref{M}$ are all-zero, $(T^{-1})_{:,:r}\rref{M}_{:r,:}=M$, so we can set $(C_0,F_0)=((T^{-1})_{:,:r},\rref{M}_{:r,:})$.

Now suppose there are $C\in\F^{m\times r},\ F\in\F^{r\times n}$ s.t. $CF=M$.
Since $\rowspan{M}=\rowspan{CF}\subseteq\rowspan{F}$ and $\rank{F}\le r=\rank{M}$, we have $\rowspan{M}=\rowspan{F}$ and $\rank{F}=r$.
Since row reduction preserves rowspan, $\rowspan{\rref{M}}=\rowspan{\rref{F}}$, so by Lemma \ref{unique-rref}, $\rref{F}=\rref{M}_{:r,:}=F_0$.

By Lemma \ref{fast-rankfac} again, $\exists X\in\GL{r}{\F}$ s.t. $XF=F_0$. Solving for $F$ and substituting it in the original equation $CF=M$ yields $(CX^{-1})F_0=M$. Since $\rank{F_0}=r$, its rows are linearly independent, so there is at most one solution for $CX^{-1}$. Since $C_0$ is a solution, $CX^{-1}=C_0$. Thus, $C=C_0X$ and $F=X^{-1}F_0$.
\end{proof}

Since invertible matrices form a group, it follows that all pairs of full-rank factorizations of $M$ are reachable from each other via the transformation $(C,F)\mapsto(CX,X^{-1}F)$, so the lemma stays true when $(C_0,\ F_0)$ is an arbitrary full-rank factorization of $M$, not just the one where $F_0=\rref{M}_{:r,:}$.

Lastly, we bound the number of $m\times n$ matrices with rank $\le 1$ and enumerate them as efficiently as possible (up to $\poly(R)$ factors).
\begin{lemma}
\label{rank1-enum}
There are $\fexp{m+n}{1}$ many matrices $M\in\F^{m\times n}$ s.t. $\rank{M}\le 1$,
and enumerating all of them takes $\fexp{m+n}{1}$ time.
\end{lemma}
\begin{proof}
WLOG $\rank{M}=1$, otherwise $M=O$ (only one possibility).
Then there is exactly one way of expressing $M$ as $u\times v$ with the requirement that $u$ is canonically normalized, so we can iterate over all such $u\in\F^m$ and all nonzero $v\in\F^n$ in $\frac{|\F|^m-1}{|\F|-1}\paren{|\F|^n-1}=\fexp{m+n}{1}$ time.
\end{proof}

Note that the bound on the number of such $M$ is a corollary of Lemma \ref{fullrank}.

\subsubsection{Rank-$r$ factorization summands of a rank-$r$ matrix}
\label{fullrank-enum-proof}
We now prove Lemma \ref{fullrank-enum}.
Using the same matrices as defined in Lemma \ref{fullrank}, we have that for any $i$, $C_{:,i}F_{i,:}=C_0 \paren{X_{:,i} (X^{-1})_{i,:}} F_0$. Since the paranthesized part is a $r\times r$ matrix of rank $\le 1$, by Lemma \ref{rank1-enum} the whole expression has $\fexp{2r}{1}$ many possibilities.

\subsubsection{Rank-$r$ factorization summands of a rank-$(r-1)$ matrix}
\label{add1rank-enum-proof}
To prove Lemma \ref{add1rank-enum}, we further derive a rank bound on the factors of a matrix factorization.

\begin{lemma}
\label{rankfac-ineq}
Let $(U,\ V)$ be a rank-$r$ factorization of $M\in\F^{m\times n}$ and $q:=\rank{M}$. Then $\rank{U}+\rank{V}\le r+q$.
\end{lemma}
\begin{proof}
By Lemma \ref{fast-rankfac}, $\exists F\in\GL{m}{\F}$ s.t. $FM=\rref{M}$, so $(FM)_{q:,:}=(FU)_{q:,:}V=O$.
Transposing this equation shows that each column of $((FU)_{q:,:})^\intercal$ is in $\nullspace{V^\intercal}$, so $\rank{((FU)_{q:,:})^\intercal}=\rank{(FU)_{q:,:}}=\dim{\nullspace{V^\intercal}}$.
By the rank-nullity theorem, $\dim{\nullspace{V^\intercal}}=r-\rank{V^\intercal}=r-\rank{V}$.
Adding back the first $q$ rows of $FU$, we have that $\rank{FU}\le r+q-\rank{V}$. Since left- (or right-) multiplying by an invertible matrix does not change matrix rank, $\rank{U}+\rank{V}\le r+q$.
\end{proof}

When $q=r-1$, we get a stronger result:

\begin{lemma}
Let $(C,\ F)$ be a rank-$r$ factorization of $M\in\F^{m\times n}$.
Suppose $\rank{M}=r-1$, and let $(C_0,\ F_0)$ be a rank-$(r-1)$ factorization of $M$.
Then $(\exists G\in\F^{(r-1)\times r} \textrm{ s.t. } C=C_0G)$ or $(\exists H\in\F^{r\times (r-1)} \textrm{ s.t. } F=HF_0)$.
\end{lemma}
\begin{proof}
By Lemma \ref{rankfac-ineq}, $\rank{C}+\rank{F}\le 2r-1$. Since $\rank{C}\ge r-1$ and $\rank{F}\ge r-1$, we must have $\rank{C}=r-1$ or $\rank{F}=r-1$. 

If $\rank{C}=r-1$, then there exists a rank-$(r-1)$ factorization $(C',\ A)$ of $C$. Then $(C',\ AF)$ is a rank-$(r-1)$ factorization of $M$, so by Lemma \ref{fullrank} $\exists X\in\GL{r-1}{\F}$ s.t. $(C',\ AF)=(C_0X,\ X^{-1}F_0)$, and thus $C=C_0 (XA)$.

Otherwise, $\rank{F}=r-1$, and a similar result can be obtained for $F$ by transposing the previous argument.
\end{proof}

If such $G$ exists, then for any $i$, $C_{:,i}F_{i,:}=C_0(G_{:,i}F_{i,:})$; otherwise, $C_{:,i}F_{i,:}=(C_{:,i}H_{i,:})F_0$. In each case, the paranthesized part is a matrix of rank $\le 1$, with shape $(r-1)\times n$ and $m\times (r-1)$ respectively. By Lemma \ref{rank1-enum}, $C_{:,i}F_{i,:}$ has a total of $\fexp{m+r-1}{1}+\fexp{n+r-1}{1}$ many possibilities.

Note that when $q=r-2$, we can no longer upper-bound the number of possibilities of $C_{:,i}F_{i,:}$ for any individual $i$ beyond the na\"ive bound of $\fexp{m+n}{1}$, since if $(U,\ V)$ is a full-rank factorization of $M$, then $\paren{\MA{c|c|c}{U&u&u},\ \MA{c}{V\\\hline v\\\hline -v}}$ is a rank-$r$ factorization of $M$ for arbitrary vectors $u\in\F^{m\times 1},\ v\in\F^{1\times n}$. However, it may still be possible to give a nontrivial bound on the number of possibilities of lists $\bracket{C_{:,i}F_{i,:}}_{0\le i<\ell}$ w.r.t. $\ell$ for $\ell>1$.

\end{appendices}

\begin{thebibliography}{999}
    \bibitem{adaflip} Arai, Y., Ichikawa, Y., \& Hukushima, K. (2024).
    Adaptive Flip Graph Algorithm for Matrix Multiplication.
    \url{https://arxiv.org/abs/2312.16960}
    
    \bibitem{survey} Bl\"aser, M. (2013).
    Fast Matrix Multiplication.
    \textit{Theory of Computing.}
    \url{https://theoryofcomputing.org/articles/gs005/}

    \bibitem{cpdqz}
    Evert, E., Vandecappelle, M., \& Lathauwer, L. D. (2022).
    Canonical Polyadic Decomposition via the generalized Schur decomposition.
    \url{https://arxiv.org/abs/2202.11414}

    \bibitem{alphatensor}
    Fawzi, A., Balog, M., Huang, A., Hubert, T., Romera-Paredes, B., Barekatain, M., Novikov, A., Ruiz, F. J. R., Schrittwieser, J., Swirszcz, G., Silver, D., Hassabis, D., \& Kohli, P. (2022).
    Discovering faster matrix multiplication algorithms with reinforcement learning.
    \textit{Nature, 610}, 47-53.
    \url{https://doi.org/10.1038/s41586-022-05172-4}

    \bibitem{hardness}
    H\r{a}stad, J. (1990).
    Tensor rank is NP-complete.
    \textit{Journal of Algorithms, 11}(4), 644-654.
    \url{https://doi.org/10.1016/0196-6774(90)90014-6}

    \bibitem{flip}
    Kauers, M., \& Moosbauer, J. (2022).
    Flip graphs for matrix multiplication.
    \url{https://arxiv.org/abs/2212.01175}

    \bibitem{flip2}
    Kauers, M., \& Moosbauer, J. (2023).
    Some New Non-Commutative Matrix Multiplication Algorithms of Size $(n,m,6)$.
    \url{https://arxiv.org/abs/2312.16960}

    \bibitem{fpt-approx}
    Mahankali, A. V., Woodruff, D. P., Zhang, Z. (2023).
    Near-Linear Time and Fixed-Parameter Tractable Algorithms for Tensor Decompositions.
    \url{https://arxiv.org/abs/2207.07417}

    \bibitem{simul-diag}
    Moitra, A. (2014).
    Algorithmic Aspects of Machine Learning.
    \url{https://www.cs.cmu.edu/~ninamf/courses/806/lect1111_tensor.pdf}

    \bibitem{smirnov}
    Smirnov, A. V. (2013).
    The Bilinear Complexity and Practical Algorithms for Matrix Multiplication.
    \url{https://cs.uwaterloo.ca/~eschost/Exam/Smirnov.pdf}
\end{thebibliography}
\end{document}